\documentclass[letterpaper, 10 pt, conference]{IEEEconf}
\IEEEoverridecommandlockouts

\usepackage{cite}
\usepackage{graphicx}
\usepackage{amsmath, amsfonts, amssymb, bm}
\usepackage{enumerate}
\usepackage{xcolor}
\usepackage{flushend}

\allowdisplaybreaks[1]

\newcommand{\Expectation}{\mathbb{E}}
\renewcommand{\Re}{\mathbb{R}}
\newcommand{\diag}{\mathtt{diag}}

\newtheorem{theorem}{Theorem}
\newtheorem{lemma}{Lemma}
\newtheorem{definition}{Definition}
\newtheorem{remark}{Remark}

\begin{document}

\title{Adaptive
Dual Covariance Steering
with Active Parameter Estimation
}

\author{Jacob W. Knaup \and Panagiotis Tsiotras
\thanks{J. Knaup is with the School of Interactive Computing, College of Computing and
the Institute for Robotics and Intelligent Machines,
Georgia Institute of Technology,  Atlanta, GA 30332--0250, USA (e-mail: jacobk@gatech.edu)}
\thanks{P. Tsiotras is with the School of Aerospace Engineering and the
 Institute for Robotics and Intelligent Machines,
Georgia Institute of Technology,  Atlanta, GA 30332--0150, USA (e-mail: tsiotras@gatech.edu)}
}

\maketitle

\begin{abstract}
This work examines the optimal covariance steering problem for systems subject to unknown parameters that enter multiplicatively with the state and control, in addition to additive disturbances. 
In contrast to existing works, the unknown parameters are modeled as random variables and are estimated online.
This work proposes the utilization of recursive least squares estimation for efficient parameter identification. 
A dual control problem is formulated in which the effect of the planned control policy on the parameter estimates is modeled and optimized for.
The parameter estimates are then used to modify the pre-computed control policy online in an adaptive control fashion.
Finally, the proposed approach is demonstrated in a vehicle control example with closed-loop parameter identification.
\end{abstract}

\section{Introduction}

Optimal covariance steering deals with the problem of steering a stochastic system from a given initial mean and covariance to a terminal mean and covariance while minimizing a given performance criterion~\cite{goldshtein2017finite, chen2015optimal1, chen2015optimal2, chen2018optimal}. 
Covariance steering has been successfully applied to a variety of practical problems including path planning~\cite{okamoto2019optimal, zheng2022belief}; differential games~\cite{chen2019covariance}; entry, descent, and landing for space systems~\cite{ridderhof2018uncertainty, goyal2021optimal}; and orbital operations~\cite{ridderhof2022chance, benedikter2022convex, benedikter2022covariance}.
While originally the problem was formulated for Gaussian distributions---in which case the initial and terminal distributions were fully defined by the initial and terminal means and covariances---recent results have moved towards more general distributions with finite moments~\cite{liu2022generic, ridderhof2019nonlinear, chen2021covariance}.
Moreover, in addition to additive disturbances, recent works have also considered systems subject to multiplicative disturbances~\cite{liu2022optimal, knaup2023computationally, balci2023covariance} and unknown parameters~\cite{knaup2023covariance}.

While control policies are typically designed for idealized models, real physical systems, such as the vehicles and spacecraft considered in~\cite{ridderhof2018uncertainty, goyal2021optimal, ridderhof2022chance, benedikter2022convex, benedikter2022covariance, knaup2023computationally, knaup2023covariance}, depend on a variety of physical parameters that can be difficult to determine exactly.
Typically, the parameter identification is performed offline and then the nominal identified parameter is taken as the ground truth for the purposes of control design.
Although adaptive control extends this paradigm by performing the identification online, these methods, nonetheless, still only use the nominal estimate for the purposes of control design~\cite{soloperto2019dual, filatov2000survey, mesbah2018stochastic}.
This, however, can lead to over-confidence in the control policy and poor performance if the parameter estimate is inaccurate~\cite{unbehauen2000adaptive}.
In~\cite{knaup2023covariance}, we proposed an alternative paradigm in which the estimated parameter \emph{distribution} is used for the purposes of control design.
This method accounts for lingering uncertainty in the true value of the parameter after estimation, and designs a control policy that is robust to this additional source of uncertainty. 
This method also extends to the case in which parameters cannot be estimated a-priori but are known to lie in a given set.

An alternative method is to identify the parameters online in closed-loop. 
This has the advantage of enabling continuous, lifelong learning of the system's parameters and is also applicable to systems for which generating suitable test signals for parameter identification is difficult (e.g., because the system must obey safety constraints while the parameters are being identified)~\cite{unbehauen2000adaptive}.
This gives rise to the question of dual control, which addresses how to optimally control the system while performing identification to balance the performance of the system's operation and the accuracy of the resulting parameter estimates~\cite{unbehauen2000adaptive, filatov2000survey}. 
In general, there is a trade-off between these objectives as, typically, introducing increased excitation leads to better parameter estimates but at the expense of the performance and safety of the system~\cite{bar1974dual, mesbah2018stochastic}.

This work builds on our recent results in~\cite{knaup2023covariance} dealing with covariance steering for systems subject to unknown parameters and additive disturbances. 
The present work examines linear time-varying systems subject to additive disturbances. 
Additionally, we consider that these systems are subject to unknown parameters, modeled as random variables drawn from a known distribution which may enter the system multiplicatively with the state and control as well as additively. 
In contrast to~\cite{knaup2023covariance}, which designs a fixed optimal control policy for the given \emph{distribution} of parameters, the present work proposes to estimate the parameters online and design an optimal control policy which \emph{adapts} to the parameter \emph{estimates}.

Specifically, we formulate an optimal dual control problem which models the effect of the control policy on the uncertainty of the parameter estimates, and designs an optimal feedback policy which adapts to the parameter estimates as they converge.
Although the feedback policy is designed offline, the adaptation takes place online as the data is collected; thus, feedback is applied both with respect to the state measurement and the parameter estimate.
The parameter estimation is performed efficiently using recursive least squares estimation which avoids the need for solving an optimization problem online.
Therefore, the resulting policy may be efficiently implemented on embedded systems with extremely limited computational resources.
This work represents the first introduction of online learning and parameter estimation into the covariance steering literature and, in particular, presents the first dual control formulation of the covariance steering problem.

\section{Problem Formulation}

Consider an affine time-and-parameter-varying system given by 
\begin{subequations}
    \begin{align} \label{eq:lpv_sys}
        x_{k+1} &= A_{k}(p) x_k + B_{k}(p) u_k + D_{k} w_k + r_{k}(p),
    \end{align}
    where,
    \begin{align}
        A_{k}(p) = \sum_{j=1}^{n_p} A_{j, k} p^{j}, &\quad
        B_{k}(p) = \sum_{j=1}^{n_p} B_{j, k} p^{j}, \\
        r_{k}(p) &= \sum_{j=1}^{n_p} r_{j, k} p^{j},
    \end{align}
\end{subequations}
and where $p = [p^{1}, \dots, p^{n_p}]^\top \in \Re^{n_p}$, $x_k \in \Re^{n_x}$, and $u_k \in \Re^{n_u}$.
Let the parameter $p \sim \mathcal{P}$, such that $\Expectation[p] = \bar{p}$ and $\Expectation[(p - \bar{p})(p - \bar{p})^\top] = P \succeq \bm{0}_{n_p \times n_p}$, 
and let $w_k \sim \mathcal{W}$ such that $\Expectation[w_k] = \mathbf{0}_{n_w}$, $\Expectation[w_k w_k^\top] = I_{n_w}$, and $\Expectation[w_{k_1} w_{k_2}^\top] = 0_{n_w \times n_w}$ for all $k_1 \neq k_2$. 
It is assumed that $A_{j, k}$, $B_{j, k}$, $D_{k}$, and $r_{j, k}$ are known for all $j = 1, \dots, n_p$ and $k = 0, \dots, N-1$.
The system dynamics are therefore affine in the parameters, and thus can also be written as
\begin{subequations}
    \begin{align} 
        x_{k+1} &= \Gamma_{k}(x_k, u_k)p + D_{k} w_k,
    \end{align}
    where,
    \begin{align}
        % \bar{\Gamma}(x, u) &= A_0 x + B_0 u, \\
        \Gamma_{k}(x, u) &= [A_{1, k} x_k + B_{1, k} u_k + r_{1, k}, \dots, \nonumber\\
        &\qquad A_{n_p, k} x_k + B_{n_p, k} u_k + r_{n_p, k}].
    \end{align}
\end{subequations}

Let the initial conditions be given by $x_0 \sim \mathcal{X}_0$, where $\Expectation[x_0] = \mu_I$ and $\Expectation[(x_0 - \mu_I)(x_0 - \mu_I)^\top] = \Sigma_I$, for $\mu_I \in \Re^{n_x}$ and $\Sigma_I \succeq \bm{0}_{n_x \times n_x}$.
We wish to steer \eqref{eq:lpv_sys} to a given final mean $\mu_F \in \Re^{n_x}$ at time $N$ and regulate the terminal covariance with respect to $\Sigma_F \succ \bm{0}_{n_x \times n_x}$, such that 
\begin{align*}
    \Expectation[x_N] = \mu_F, \quad \Expectation[(x_N - \Expectation[x_N])(x_N - \Expectation[x_N])^\top] \preceq \Sigma_F,
\end{align*}
while minimizing the cost function
\begin{subequations}
\begin{align}\label{eq:cost_function}
    J(\mu_I, \Sigma_I; u_0, \dots, u_{N-1}) &= \Expectation\left[\sum_{k=0}^{N-1} \ell_{k}(x_k, u_k) \right],
\end{align}
% where $\ell_{k}(x_k, u_k) \in \mathcal{C}^2 : \Re^{n_x} \times \Re^{n_u} \rightarrow \Re$ is twice continuously differentiable for all $k = 0, \dots, N-1$.
where,
\begin{align}
    \ell_k(x_k, u_k) &= x_k^\top Q_k x_k + u_k^\top R_k u_k,
\end{align}
and $Q_k \succeq \bm{0}_{n_x \times n_x}$, $R_k \succ \bm{0}_{n_u \times n_u}$.
\end{subequations}

Thus, the problem may be summarized by designing an $N$-step control policy $\bm{\pi}^{N} = \{\pi_0, \dots, \pi_{N-1} \}$ that solves the problem
\begin{subequations} \label{prob:fi_cs}
    \begin{align}
        & \min_{\bm{\pi}^{N}} J_N(\bm{\pi}^{N}) = \Expectation\left[\sum_{k=0}^{N-1} \ell_{k}(x_k, u_k) \right], \\
        &\text{subject to} \nonumber\\
        & x_0 \sim \mathcal{X}_{0}, \quad w_k \sim \mathcal{W}, \\
        & x_{k+1} = A_{k}(p_\mathrm{gt}) x_k + B_{k}(p_\mathrm{gt}) u_k + D_{k} w_k + r_{k}(p_\mathrm{gt}), \label{const:dynamics} \\
        & u_k = \pi_k(x_k), \\
        & \Expectation[x_N] = \mu_F, \quad \Expectation[(x_N - \mu_F)(x_N - \mu_F)^\top] \preceq \Sigma_F,
    \end{align}
\end{subequations}
for $k = 0, 1, \dots, N-1$, and where the dynamics \eqref{const:dynamics} evolve according to the particular, unknown realization of $p_\mathrm{gt} \sim \mathcal{P}$.

\section{Covariance Steering Controller Design}

\subsection{Existing Methods}

While Problem~\eqref{prob:fi_cs} represents the ideal problem we would like to solve using the true parameter value, it is not tractable as the realization of $p_{gt} \sim \mathcal{P}$ is not known ahead of time, nor can it be observed directly. 
One alternative is to solve the \textit{certainty equivalence} problem given by
\begin{subequations} \label{prob:ce_cs}
    \begin{align}
        & \min_{\bm{\pi}^{N}} J_N(\bm{\pi}^{N}) = \Expectation\left[\sum_{k=0}^{N-1} \ell_k(x_k, u_k) \right], \\
        &\text{subject to} \nonumber\\
        % & \Expectation[x_0] = \mu_I, \quad \Expectation[(x_0 - \mu_I)(x_0 - \mu_I)^\top] = \Sigma_I, \\
        & x_0 \sim \mathcal{X}_{0}, \quad w_k \sim \mathcal{W}, \\
        & x_{k+1} = A_k(\bar{p}) x_k + B_k(\bar{p}) u_k + D_k w_k + r_k(\bar{p}), \\
        & u_k = \pi_k(x_k), \\
        & \Expectation[x_N] = \mu_F, \quad \Expectation[(x_N - \mu_F)(x_N - \mu_F)^\top] \preceq \Sigma_F,
    \end{align}
\end{subequations}
for $k = 0, 1, \dots, N-1$, and where the expected value of $p$, given by $\bar{p}$, is used in place of the true value, $p_{\mathrm{gt}}$.
This is the case, for example, when system identification is performed offline and then the nominal estimated parameter values are taken as the true values for the purpose of control design.
However, this approach does not consider the uncertainty in the parameters and will fail to meet the terminal constraints in practice if $p_{\mathrm{gt}} \neq \bar{p}$.

Another alternative, similar to the approach used in \cite{knaup2023covariance}, is to solve the stochastic problem given by 
\begin{subequations} \label{prob:r_cs}
    \begin{align}
        & \min_{\bm{\pi}^{N}} J_N(\bm{\pi}^{N}) = \Expectation\left[\sum_{k=0}^{N-1} \ell_k(x_k, u_k) \right], \\
        &\text{subject to} \nonumber\\
        % & \Expectation[x_0] = \mu_I, \quad \Expectation[(x_0 - \mu_I)(x_0 - \mu_I)^\top] = \Sigma_I, \\
        & x_0 \sim \mathcal{X}_{0}, \quad w_k \sim \mathcal{W}, \quad p \sim \mathcal{P}, \\
        & x_{k+1} = A_k({p}) x_k + B_k({p}) u_k + D_k w_k + r_k({p}), \\
        & u_k = \pi_k(x_k), \\%\quad w_k \sim \mathcal{W}, \quad p \sim \mathcal{P}, \\
        & \Expectation[x_N] = \mu_F, \quad \Expectation[(x_N - \mu_F)(x_N - \mu_F)^\top] \preceq \Sigma_F,
    \end{align}
\end{subequations}
for $k = 0, 1, \dots, N-1$,
and where the known distribution of $p$, given by $\mathcal{P}$, is used to propagate the uncertainty in the state dynamics.
In contrast to Problem~\eqref{prob:ce_cs}, the expectations in Problem~\eqref{prob:r_cs} are also taken over $p \sim \mathcal{P}$ in order to account for the additional source of uncertainty.
This approach is robust to the uncertainty and will meet the terminal constraints in practice. 
However, this approach is overly conservative, as it must design a single state-feedback control policy which statistically performs well for all possible parameter realizations.

A more optimal policy would \emph{adapt} to the parameter realization to approximate the solution to \eqref{prob:fi_cs}. 
To this end, introduce the control policy
\begin{align} \label{eq:adaptive_control_policy}
    u_k &= \rho_k(x_k, p),
\end{align}
where $\rho_k(\cdot, \cdot) : \Re^{n_x} \times \Re^{n_p} \rightarrow \Re^{n_u}$ for all $k = 0, 1, \dots, N-1$. 
Such a control policy is referred to as \emph{adaptive}, which we define below.
\begin{definition}[Adaptive Control \cite{unbehauen2000adaptive, seborg1986adaptive, isermann1982parameter}]\label{def:adaptive}
    An adaptive control system is one that automatically changes the parameters of the control policy in order to improve performance by adjusting its behavior to the properties of the system to be controlled.

\end{definition}
% 
% We then wish to solve for the optimal sequence of policies given by $\bm{\rho}^N = \{\rho_0(\cdot, \cdot), \rho_1(\cdot, \cdot), \dots, \rho_{N-1}(\cdot, \cdot) \}$.
% \textcolor{red}{DELETE THIS}
% The optimal sequence is given by the solution to 
% \begin{subequations} \label{prob:a_cs}
%     \begin{align}
%         & \min_{\bm{\rho}^{N}} J_N(\bm{\rho}^{N}) = \Expectation\left[\sum_{k=0}^{N-1} \ell_k(x_k, u_k) \right], \\
%         &\text{subject to} \nonumber\\
%         & \Expectation[x_0] = \mu_I, \quad \Expectation[(x_0 - \mu_I)(x_0 - \mu_I)^\top] = \Sigma_I, \\
%         & x_{k+1} = A_k({p}) x_k + B_k({p}) u_k + D_k w_k + r_k({p}), \\
%         & u_k = \rho_k(x_k, p), \quad w_k \sim \mathcal{W}, \quad p \sim \mathcal{P}, \\
%         & \Expectation[x_N] = \mu_F, \quad \Expectation[(x_N - \mu_F)(x_N - \mu_F)^\top] \preceq \Sigma_F,
%     \end{align}
% \end{subequations}
% for $k = 0, 1, \dots, N-1$.
% However, Problem~\eqref{prob:a_cs} 
% is intractable because the control policy depends on the unobservable parameter realizations. Instead, we propose a control policy which adapts to the \emph{estimated} parameters. 
% 
However, policy \eqref{eq:adaptive_control_policy} is inadmissible because the control policy depends on the unobservable parameter realization $p$. 
Instead, we propose a control policy which adapts to the \emph{estimated} parameter. 

\subsection{Proposed Approach}

We propose to estimate the unknown parameters using the method of weighted least squares given by
% \begin{subequations}
    \begin{align} \label{prob:wbls}
        &\min_{{p}\,\in\,\Re^{n_p}} ~ \gamma^{k} ({p} - \bar{p})^\top P^{-1} ({p} - \bar{p}) \\
        &+ \sum_{t=0}^{k-1} \gamma^{k-t-1} (x_{t+1} - \Gamma_{k}(x_t, u_t){p})^{\top} (x_{t+1} - \Gamma_{k}(x_t, u_t){p}), \nonumber 
    \end{align}
% \end{subequations}
for $k = 0, 1, \dots, N-1$, and where $P$ is the covariance of $\mathcal{P}$ and $\gamma \in (0, 1]$.
If $\gamma < 1$, then $\gamma$ acts as an exponentially decaying weight and reduces the influence of the prior and older measurements over time, ensuring the estimated parameter fits the most recent data.
The case of $\gamma = 1$ is also permitted, and in this case, all observations and the prior are weighted equally.
The inclusion of the prior, which is constructed using the first and second moments of the parameter distribution, ensures that Problem~\eqref{prob:wbls} has a unique solution regardless of the available data $x_0, \dots, x_k$, and $u_0, \dots, u_{k-1}$.
The optimal parameter estimate which solves \eqref{prob:wbls} at time $k$ is given by $\hat{p}_k \in \Re^{n_p}$ and depends on the observations available up to and including time $k$ as well as on the first and second moments of the parameter distribution.

We then employ the adaptive control policy \eqref{eq:adaptive_control_policy}, and use the estimated parameter $\hat{p}_k$ in place of the true parameter $p$. 
Thus, the control is given by 
\begin{align}\label{eq:dual_control_policy}
    u_k = \rho_k(x_k, \hat{p}_k),
\end{align}
which clearly possesses the adaptive property given in Definition~\ref{def:adaptive}.
Additionally, we wish to design the policy \eqref{eq:dual_control_policy} in such a way that it also possesses the \emph{dual} control properties given in Definition~\ref{def:dual} below,
\begin{definition}[Adaptive Dual Control \cite{unbehauen2000adaptive, filatov2000survey, bar1974dual, mesbah2018stochastic}] \label{def:dual}
    A dual control system is one that operates under model uncertainty and incorporates this uncertainty into the control strategy such that the control signal has the dual properties:
    \begin{enumerate}[(i)]
        \item The system tracks the desired reference value and obeys the constraints in the presence of the model uncertainty.
        \item The system is excited in order to improve the estimation (by reducing the higher order moments of the parameter error) so that the quality of the adaptive control policy using the parameter estimate can be improved considerably in future time intervals.
    \end{enumerate}
    % (i) the system output cautiously tracks the desired reference value, and 
    % (ii) it excites the system to improve the estimation (by reducing the higher order moments of the parameter error) so that the quality of adaptive control policy using the parameter estimate can be improved considerably in future time intervals.

    % A control system operating under conditions of uncertainty of the controller that provides the desired system performance by changing its parameters and/or structure in order to reduce the uncertainty and to improve the operation of the desired system by incorporating the existing uncertainty into the control strategy with the control signal having the properties: (i) it follows the control goal and (ii) excites the plant to improve the estimation, is an adaptive dual control system \cite{unbehauen2000adaptive}.

    % A control input is said to have dual control effect if it can affect, with nonzero probability, at least one rth-order central moment of a state variable ($r \geq 2$) \cite{bar1974dual, mesbah2018stochastic}.

    % The control signal should ensure that (i) the system output cautiously tracks the desired reference value, and (ii) it excites the plant sufficiently to accelerate the parameter estimation process, so that the quality of adaptive controllers designed on the parameter estimation process can be improved considerably in future time intervals. \cite{filatov2000survey}
\end{definition}

We are now ready to introduce the adaptive dual covariance steering problem, which is given as
\begin{subequations} \label{prob:d_cs}
    \begin{align}
        & \min_{\bm{\rho}^{N}} J_N(\bm{\rho}^{N}) = \Expectation\left[\sum_{k=0}^{N-1} \ell_k(x_k, u_k) \right], \\
        &\text{subject to} \nonumber\\
        % & \Expectation[x_0] = \mu_I, \quad \Expectation[(x_0 - \mu_I)(x_0 - \mu_I)^\top] = \Sigma_I, \\
        & x_0 \sim \mathcal{X}_{0}, \quad w_k \sim \mathcal{W}, \quad p \sim \mathcal{P}, \\
        & x_{k+1} = A_k({p}) x_k + B_k({p}) u_k + D_k w_k + r_k(p), \\
        & u_k = \rho_k(x_k, \hat{p}_k), \label{const:est_fb}\\
        & \hat{p}_k = \mathrm{argmin}_{{p}} ~ \gamma^{k} ({p} - \bar{p})^\top P^{-1} ({p} - \bar{p}) \label{const:lse}\\
        &+ \sum_{t=0}^{k-1} \gamma^{k-t-1} (x_{t+1} - \Gamma_{k}(x_t, u_t){p})^{\top} (x_{t+1} - \Gamma_{k}(x_t, u_t){p}), \nonumber\\
        & \Expectation[x_N] = \mu_F, \quad \Expectation[(x_N - \mu_F)(x_N - \mu_F)^\top] \preceq \Sigma_F,
    \end{align}
\end{subequations}
for $k = 0, 1, \dots, N-1$, and where $\bm{\rho}^N = \{\rho_0(\cdot, \cdot), \dots, \rho_{N-1}(\cdot, \cdot) \}$, ${p} \in \Re^{n_p}$, and $\gamma \in (0, 1]$. 

Most adaptive control policies are based on the separation of parameter estimation and controller design and only identify a model passively as a byproduct.
In such cases, the control law is designed using the estimated parameters as an exact representation for the system in a certainty equivalence fashion, without accounting for the uncertainty of estimation \cite{soloperto2019dual, filatov2000survey}.
Problem~\eqref{prob:d_cs}, on the other hand, ensures the constraints are satisfied statistically for any parameter realization $p \sim \mathcal{P}$, regardless of the estimate.
In contrast to Problem~\eqref{prob:r_cs} though, Problem~\eqref{prob:d_cs} also accounts for the online estimation of the parameter and adapts the control policy to the parameter estimate  $\hat{p}_k$.
Thus, the solution is an adaptive \emph{dual} control policy, as per Definition~\ref{def:dual}.
% This assertion is made more concrete in Theorem~\ref{thm:mc_dual}.
% 

\subsection{Tractable Formulation}

The solution to Problem~\eqref{prob:d_cs} depends on optimizing over the space of $N$ arbitrary (infinite-dimensional) control policies which each depend on the estimated parameters at that time step. 
The estimated parameters, in-turn, depend on the solution of $N$ quadratic programs. 
We address these computational issues by first introducing an analytical, recursive solution to the quadratic program~\eqref{prob:wbls}. 
Then, we introduce a finite-dimensional approximation of $\rho_k(\cdot, \cdot)$.

% \textcolor{red}{DELETE THIS}
% \begin{proposition}[\!\!\cite{islam2019recursive, isermann2011identification}]
%     Provided $P \succ 0$, Problem~\eqref{prob:wbls} has a unique solution given by 
%     \begin{align}
%         \hat{p}_k &= (\Phi_{0:k-1}^\top \Phi_{0:k-1})^{-1} (\Phi_{0:k-1}^\top X_{1:k}), \label{eq:wbls_sol}
%     \end{align}
%     where 
%     % \begin{subequations}
%         \begin{align}
%             \Phi_{k_1:k_2} = \begin{bmatrix}
%                 \Gamma_{k_1}(x_{k_1}, u_{k_1}) \\
%                 \Gamma_{k_1 + 1}(x_{k_1 + 1}, u_{k_1 + 1}) \\
%                 \vdots \\
%                 \Gamma_{k_2}(x_{k_2}, u_{k_2})
%             \end{bmatrix}
%             \quad 
%             X_{k_1:k_2} = \begin{bmatrix}
%                 x_{k_1} \\
%                 x_{k_1 + 1} \\
%                 \vdots \\
%                 x_{k_2}
%             \end{bmatrix}
%         \end{align}
%     % \end{subequations}
% \end{proposition}
% 
The solution to Problem~\eqref{prob:wbls} requires storing the entire history of the observations, which grows with $k$, making the approach impractical for large time horizons.
To resolve this issue, we utilize recursive least squares, given by the following lemma.
\begin{lemma}[Recursive Least Squares \cite{islam2019recursive}] \label{lem:rls}
    The solution to Problem~\eqref{prob:wbls} is given by
    \begin{subequations}\label{eq:rls}
        \begin{align}
            &\hspace{-1mm} \hat{p}_{k+1} = \hat{p}_k + P_{k+1} \Gamma_{k}(x_k, u_k)^{\top} (x_{k+1} - \Gamma_{k}(x_k, u_k) \hat{p}_{k}), \\
            &\hspace{-1mm} P_{k+1} = \frac{1}{\gamma} P_{k} - \frac{1}{\gamma}P_{k} \Gamma_{k}(x_k, u_k)^{\top} \nonumber\\
            &\;\; (\gamma I + \Gamma_{k}(x_k, u_k) P_{k} \Gamma_{k}(x_k, u_k)^{\top})^{-1} \Gamma_{k}(x_k, u_k) P_{k},
        \end{align}
    \end{subequations}
    for $k = 0, 1, \dots, N-1$, where $\hat{p}_0 = \bar{p}$ and $P_0 = P$.
\end{lemma}
\begin{proof}
    The result follows immediately from Theorem~1~of~\cite{islam2019recursive}.
\end{proof}
Lemma~\ref{lem:rls} allows for the parameter to be updated using only the two most recent data points. 
Thus, the sizes of the matrices and vectors in \eqref{eq:rls} are fixed and do not depend on $k$.
For more details on recursive least squares, see, for example, \cite{liu2013convergence, isermann2011identification}.

Having rendered the parameter estimation tractable, we turn to the control policy parameterization.
Since optimizing over arbitrary feedback policies amounts to an intractable infinite dimensional optimization problem, in practice, prior works, such as \cite{goldshtein2017finite, rapakoulias2023discrete, knaup2023covariance}, introduce to Problems~\eqref{prob:fi_cs},~\eqref{prob:ce_cs},~\eqref{prob:r_cs} the affine state-feedback policy parameterization given by 
\begin{align} \label{eq:affine_state_feedback}
    \pi_k(x_k) &= v_k + L_k x_k,
\end{align}
where $v_k \in \Re^{n_u}$ and $L_k \in \Re^{n_u \times n_x}$ for $k = 0, 1, \dots, N-1$.
The affine state feedback parameterization given by \eqref{eq:affine_state_feedback} is reasonable due to the following lemma.

\begin{lemma} [\!\!\cite{liu2022optimal}] \label{lem:control_policy}
    The optimal solution for the exact covariance steering (equality constrained) version of Problem~\eqref{prob:fi_cs} is given by
    \begin{align} \label{eq:adaptive_state_feedback}
        \pi^{\ast}(x_k) = \rho^{\ast}_k(x_k, p_{\mathrm{gt}}) = \tilde{v}^{\ast}_k(p_{\mathrm{gt}}) + \tilde{L}^{\ast}_k(p_{\mathrm{gt}}) x_k,
    \end{align}
    where $\tilde{v}^{\ast}_k(\cdot) : \Re^{n_p} \rightarrow \Re^{n_u}$ and $\tilde{L}^{\ast}_k(\cdot) : \Re^{n_p} \rightarrow \Re^{n_u \times n_x}$ for $k = 0, 1, \dots, N-1$, and where
    $\tilde{v}^{\ast}_k(\cdot)$ and $\tilde{L}^{\ast}_k(\cdot)$ are given in Theorem~3 of \cite{liu2022optimal}.
    Furthermore,
    the result also holds for Problem~\eqref{prob:ce_cs}, using $\bar{p}$ in place of $p_{\mathrm{gt}}$.
\end{lemma}
\begin{proof}
    The result follows from Theorem~3~of~\cite{liu2022optimal}, with the following considerations.
    While \cite{liu2022optimal} presents the optimal policy as $u_k = K_k(x_k - \mu_k) + z_k$, one may observe that the mean sequence is deterministic and thus it may be incorporated into the feed-forward gain, to arrive at $\tilde{L}^{\ast}_k(\cdot) = K_k$ and $\tilde{v}^{\ast}_k(\cdot) = z_k - K_k \mu_k$.
    Additionally, observe the expressions for the optimal values of $z_k$ and $K_k$ given in Theorem~3~of~\cite{liu2022optimal} depend on the initial moments as well as the system matrices, which in-turn, depend on the system parameter realization, hence the mappings $\tilde{v}^{\ast}_k(p_{\mathrm{gt}})$ and $\tilde{L}^{\ast}_k(p_{\mathrm{gt}})$. 
\end{proof}

Thus, as Lemma~\ref{lem:control_policy} illustrates, the optimal control policy depends on the system parameter realization which is unknown a-priori at the time the control policy is designed in the cases of Problems~\eqref{prob:r_cs}~and~\eqref{prob:d_cs}. 
Therefore, the parameterization proposed in~\eqref{eq:adaptive_control_policy} and~\eqref{eq:dual_control_policy} is reasonable.
As optimizing over functions is intractable, we replace~\eqref{eq:dual_control_policy} with a first-order approximation of the optimal policy presented in Lemma~\ref{lem:control_policy}, given by 
% \begin{subequations}
    \begin{align} \label{eq:adaptive_affine_state_feedback}
        \rho_k(x_k, \hat{p}_k) &= v_k^0 + L_k^{0} x_k + \sum_{j=1}^{n_p} v_k^{j} \hat{p}_k^j + L_k^j \hat{p}_k^j x_k,
    \end{align}
    % where 
    % \begin{align}
    %     \tilde{v}_k(\hat{p}_k) &\triangleq v_k^0 + \sum_{j=1}^{n_p} v_k^{j} \hat{p}_k^j, \\
    %     \tilde{L}_k(\hat{p}_k) &\triangleq L_k^{0} + \sum_{j=1}^{n_p} L_k^j \hat{p}_k^j,
    % \end{align}
    where $v_k^j \in \Re^{n_u}$ and $L_k^{j} \in \Re^{n_u \times n_x}$ for $k = 0, 1, \dots, N-1$ and $j = 0, 1, \dots, n_p$.
% \end{subequations}
    
Thus, we have an alternative to Problem~\eqref{prob:d_cs}, given by 
\begin{subequations} \label{prob:d_cs_recursive}
    \begin{align}
        & \min_{\bm{v}^{N}, \bm{L}^{N}} J_N(\bm{v}^{N}, \bm{L}^{N}) = \Expectation\left[\sum_{k=0}^{N-1} \ell_k(x_k, u_k) \right], \\
        &\text{subject to} \nonumber\\
        % & \Expectation[x_0] = \mu_I, \quad \Expectation[(x_0 - \mu_I)(x_0 - \mu_I)^\top] = \Sigma_I, \\
        & x_{0} \sim \mathcal{X}_{0}, \;\; w_k \sim \mathcal{W}, \;\; p \sim \mathcal{P}, \;\; \hat{p}_0 = \bar{p}, \;\; P_0 = P, \\
        & x_{k+1} = A_k({p}) x_k + B_k({p}) u_k + D_k w_k + r_k({p}), \label{const:dr_dynamics}\\
        & u_k = v_k^0 + L_k^{0} x_k + \sum_{j=1}^{n_p} v_k^{j} \hat{p}_k^j + L_k^j \hat{p}_k^j x_k, \\
        % & \hat{p}_0 = \bar{p}, \quad P_0 = P,   \\
        &\hat{p}_{k+1} = \hat{p}_k + P_{k+1} \Gamma_{k}(x_k, u_k)^{\top} (x_{k+1} - \Gamma_{k}(x_k, u_k) \hat{p}_{k}), \label{const:rls_mean} \\
        &P_{k+1} = \frac{1}{\gamma} P_{k} - \frac{1}{\gamma}P_{k} \Gamma_{k}(x_k, u_k)^{\top} \nonumber\\
        &\;\; (\gamma I + \Gamma_{k}(x_k, u_k) P_{k} \Gamma_{k}(x_k, u_k)^{\top})^{-1} \Gamma_{k}(x_k, u_k) P_{k}, \label{const:rls_cov}\\
        & \Expectation[x_N] = \mu_F, \quad \Expectation[(x_N - \mu_F)(x_N - \mu_F)^\top] \preceq \Sigma_F,
    \end{align}
\end{subequations}
for $k = 0, 1, \dots, N-1$, and where $\bm{v}^{N} = \{v^{j}_{k}\}_{k=0, j=0}^{N-1, n_p}$ and $\bm{L}^{N} = \{L_{k}^{j}\}_{k=0, j=0}^{N-1, n_p}$.
Problem~\eqref{prob:d_cs_recursive} replaces the batch least squares estimation defined in \eqref{const:lse} with the recursive least squares estimation \eqref{const:rls_mean}-\eqref{const:rls_cov}, and Problem~\eqref{prob:d_cs_recursive} parameterizes the arbitrary feedback policies $\bm{\rho}^N$ in terms of $\bm{v}^N$ and $\bm{L}^N$.

\subsection{Sample Average Approximation}

As discussed in \cite{knaup2023covariance}, analytically propagating the state uncertainty through \eqref{const:dr_dynamics} is challenging due to the dependence between $x_k$ and $p$ for all $k > 0$. 
In \cite{knaup2023covariance}, we derived analytical expressions for the propagation of the moments for a specific class of problems subject to constant parametric uncertainty, similar to Problem~\eqref{prob:r_cs}. 
In our current work, the moment dynamics are further complicated by the online estimation of the parameters and the feedback in the control policy on the parameter estimates.
As a result, in this work, we take an alternative approach and utilize a Monte Carlo sampling-based approximation to propagate the uncertainty. 

The approximation of Problem~\eqref{prob:d_cs_recursive} is given by
\begin{subequations} \label{prob:d_cs_mc}
    \begin{align}
        & \min_{\bm{v}^{N}, \bm{L}^{N}, V} J_N^M(\bm{v}^{N}, \bm{L}^{N}) = \sum_{i=1}^{M} \sum_{k=0}^{N-1} \ell_k(x_k^{i}, u_k^{i}), \label{obj:mc}\\
        &\text{subject to} \nonumber\\
        & x_0^{i} \sim \mathcal{X}_0, \;\; w_k^{i} \sim \mathcal{W}, \;\; p^{i} \sim \mathcal{P}, \;\; \hat{p}_0^{i} = \bar{p}, \;\; P_0^{i} = P, \label{const:mc_init}\\
        & x_{k+1}^{i} = A_k(p^i) x_k^{i} + B_k({p}^{i}) u_k^{i} + D_k w_k^{i} + r_k(p^{i}), \label{const:mc_dynamics}\\
        % & u_k^{i} = \rho_k(x_k^{i}, \hat{p}_k^{i}), \\
        &  u_k^{i} = v_k^0 + L_k^{0} x_k^{i} + \sum_{j=1}^{n_p} v_k^{j} \hat{p}_k^{i, j} + L_k^j \hat{p}_k^{i, j} x_k^{i}, \label{const:mc_control}\\
        &\hat{p}_{k+1}^{i} = \hat{p}_k^{i} + P_{k+1}^{i} \Gamma_{k}(x_k^{i}, u_k^{i})^{\top} (x_{k+1}^{i} - \Gamma_{k}(x_k^{i}, u_k^{i}) \hat{p}_{k}^{i}), \label{const:mc_p_hat} \\
        &P_{k+1}^{i} = P_{k}^{i} - P_{k}^{i} \Gamma_{k}(x_k^{i}, u_k^{i})^{\top} \nonumber\\
        &\quad (I + \Gamma_{k}(x_k^{i}, u_k^{i}) P_{k}^{i} \Gamma_{k}(x_k^{i}, u_k^{i})^{\top})^{-1} \Gamma_{k}(x_k^{i}, u_k^{i}) P_{k}^{i}, \label{const:mc_p_cov} \\
        & |\sum_{i=1}^{M} x_N^{i} / M - \mu_F| \leq \Delta_{\mu}, \label{const:mc_mean}\\
        &\hspace{-1mm} |\mathtt{vec}(\sum_{i=1}^{M}  (x_N^{i} - \mu_F)(\star)^\top / M + V V^T - \Sigma_F)| \leq \Delta_{\Sigma}, \label{const:mc_cov}
    \end{align}
\end{subequations}
for $i = 1, \dots, M$, $k = 0, 1, \dots, N-1$, 
and where $(\star)$ represents repeated terms, 
$\hat{p}_{k}^{i} = [\hat{p}_{k}^{i, 1}, \dots, \hat{p}_{k}^{i, n_p} ]^\top$, $\hat{p}_{k}^{i, j} \in \Re$ for $j = 1, \dots, n_p$,
$\mathtt{vec}(\cdot)$ is the vector (flattening) operator, 
$V \in \Re^{n_x \times n_x}$,
and $\Delta_{\mu} \in \Re^{n_x}$ and $\Delta_{\Sigma} \in \Re^{n_x n_x}$ are small quantities.
For details on $V$, $\Delta_{\mu}$, and $\Delta_{\Sigma}$, refer to Remark~\ref{rem:nlp}.
Problem~\eqref{prob:d_cs_mc} approximates Problem~\eqref{prob:d_cs_recursive} using $M$ Monte Carlo trials.
The use of Monte Carlo sampling in this context is referred to as the sample average approximation (SAA), a well-established technique in the field of stochastic optimization~\cite{shapiro2021lectures}. 
For results on the convergence and feasibility of SAA, see, for example,~\cite{shapiro2021lectures, lew2022sample}.

\begin{remark} \label{rem:nlp}
    As most nonlinear programming (NLP) solvers do not support semidefinite constraints, we add the slack variable $V$ to convert the semidefinite-constrained NLP \eqref{prob:d_cs_recursive} to an equality-constrained NLP, which is then relaxed to the inequality-constrained NLP given by \eqref{prob:d_cs_mc}. 
    Note that $V V^\top \succeq \bm{0}_{n_x \times n_x}$ for all $V \in \Re^{n_x \times n_x}$, thus the relaxation using the slack variable $V$ is lossless.
    As recommended in \cite{lew2022sample}, we relax the equality constraints to inequalities allowing for a small margin of error, given by $\Delta_{\mu}$ and $\Delta_{\Sigma}$, dependent on the sample size.
\end{remark}
% 

% \begin{lemma}
%     Provided Problem~\eqref{prob:d_cs_recursive} has a solution, the SAA Problem~\eqref{prob:d_cs_mc} will also have a solution, and, moreover, the optimal solution to Problem~\eqref{prob:d_cs_mc} will be a valid solution to Problem~\eqref{prob:d_cs_recursive}.
% \end{lemma}
% \begin{proof}
% \end{proof}

% A dual control system is one that operates under model uncertainty and incorporates this uncertainty into the control strategy such that the control signal has the dual properties: (i) the system output cautiously tracks the desired reference value, and (ii) it excites the system to improve the estimation (by reducing the higher order moments of the parameter uncertainty) so that the quality of adaptive control policy using the parameter estimate can be improved considerably in future time intervals \cite{unbehauen2000adaptive, filatov2000survey, bar1974dual, mesbah2018stochastic}.

We are now ready to introduce the following theorem concerning the proposed adaptive dual covariance steering method.

\begin{theorem}\label{thm:mc_dual}
    Problem~\eqref{prob:d_cs_mc} preserves the dual control effect given in Definition~\ref{def:dual}. 
    That is, the role of the control actions in affecting the parameter estimate is accounted for in the problem formulation, and the control policy adapts to the estimated parameters so that the following dual control properties are present.
    \begin{enumerate}[(i)]
        \item The control policy minimizes the performance criteria and satisfies the terminal constraints for $p \sim \mathcal{P}$.
        \item The control policy affects the parameter uncertainty which is exploited to adapt the control policy using the improved parameter estimate at later time steps.
    \end{enumerate}
    % (i)
    % (ii)
\end{theorem}
\begin{proof}
    Satisfaction of property (i) may be seen from the dynamics \eqref{const:mc_dynamics}, in which, the future empirical state distributions $\{x_{k}^{i}\}_{k=0, i=1}^{N, M}$ depend on the parameter realizations $\{p^{i}\}_{i=1}^{M}$ sampled from the distribution $\mathcal{P}$. 
    The cost and constraint equations~\eqref{obj:mc}, \eqref{const:mc_mean}, \eqref{const:mc_cov}, then enforce the objective function and terminal constraints on the empirical state distributions $\{x_{k}^{i}\}_{k=0, i=1}^{N, M}$.
    Note that the parameter estimate only enters into the problem as a source of feedback through the control policy.
    This is in contrast to methods that enforce the constraints using the estimated parameter and which may not provide the desired performance in practice. 

    Satisfaction of property (ii) may be seen from the effect of the control on the parameter estimate in \eqref{const:mc_p_hat}-\eqref{const:mc_p_cov}.
    Specifically, note that, as per \eqref{const:mc_init}, $p_0^i = \bar{p}$ and $P_0^i = P$, for all $i=1, \dots, M$, but for $k > 0$, $p_k^i$ and $P_k^i$ depend on the previous controls $u_t^i$ where $0 \leq t < k$. 
    Therefore, the control policy may improve the parameter estimate and reduce the corresponding uncertainty. 
    The improvement in the parameter estimate at later time steps is then exploited through the control policy \eqref{const:mc_control}, which allows for more useful adaptation as $\hat{p}_k^i$ approaches the true realization of $p^i$.
\end{proof}

\section{Numerical Example}

We compare the proposed approach \eqref{prob:d_cs_mc} against the non-adaptive method given in \cite{knaup2023covariance} in an example of controlling a vehicle with uncertainty in the steering column dynamics. 
The equations of motion for the nonlinear kinematic bicycle model, shown in Fig.~\ref{fig:bike}, are given by the following 
\begin{subequations}
    \begin{align}
        % \dot{v}_x &= a_x, \\
        % \dot{\delta} &= u_{\delta}, \\
        \dot{e}_{\psi} &= \dot{\psi} - \dot{\psi}_{\mathrm{ref}} + \dot{\nu}_y / \nu_x, \\
        \dot{e}_y &= \nu_y \cos{e_{\psi}} + \nu_x \sin{e_{\psi}}, \\
        \dot{s} &= (\nu_x \cos{e_{\psi}} - \nu_y \sin{e_{\psi}}) / (1 - e_{y} \sigma),
    \end{align}
\end{subequations}
where
$\nu_y = \nu_x \delta \frac{\ell_r}{\ell_f + \ell_r}$,
$\dot{\psi} = \tan{\delta} \frac{\nu_x}{\ell_f + \ell_r}$,
and where
$e_{\psi}$ is the heading error with respect to the reference path,
$\psi$ is the vehicle's absolute heading,
$\psi_{\mathrm{ref}}$ is the heading of the reference path,
$\nu_y$ is the vehicle's lateral velocity,
$\nu_x$ is the vehicle's longitudinal velocity, 
$e_y$ is the vehicle's lateral error with respect to the reference path,
$s$ is the vehicle's longitudinal position with respect to the reference path,
$\sigma$ is the curvature of the reference path,
$\delta$ is the front steering angle,
and $\ell_f$ and $\ell_r$ are the locations of the vehicle's center of mass.
\begin{figure}[ht]
    \centering
    \includegraphics[width=0.85\linewidth]{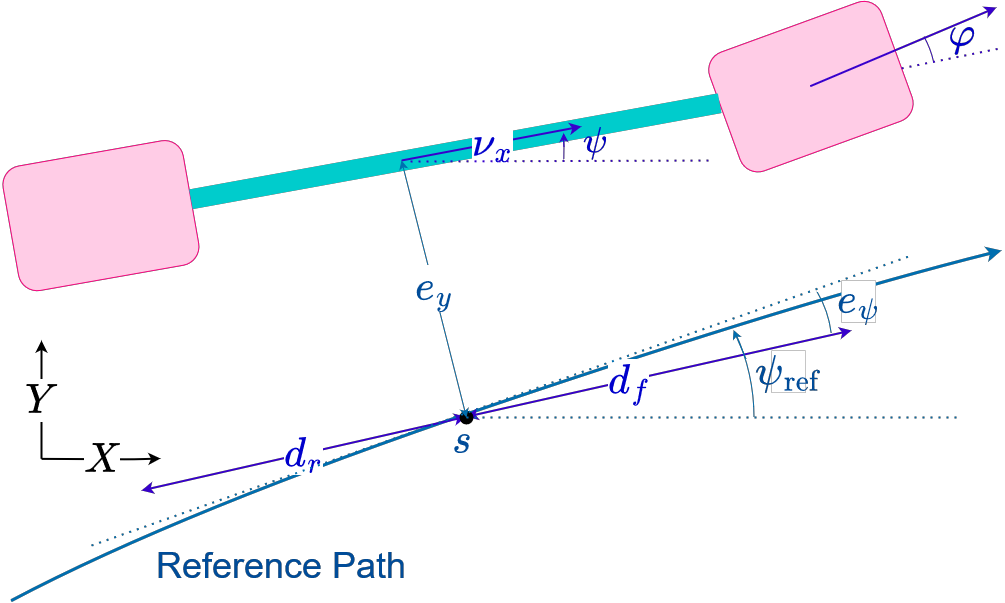}
    \caption{Kinematic bicycle model in curvilinear coordinates.}
    \label{fig:bike}
\end{figure}
Given a constant velocity $\nu_x$ and curvature $\sigma$, and assuming $\delta$ and $e_{\psi}$ remain small, a linear approximation for the lateral motion is given by
\begin{subequations}
    \begin{align}
        \dot{e}_{\psi} &= \frac{\nu_x}{\ell_f + \ell_r} \delta - \nu_x \sigma + \dot{\delta} \frac{\ell_r}{\ell_f + \ell_r}, \\
        \dot{e}_{y} &= \frac{\ell_r}{\ell_f + \ell_r} \nu_x \delta + \nu_x e_{\psi}.
    \end{align}
\end{subequations}
We define the state $x = [\delta, e_{\psi}, e_y]^\top$ and the control $u = \dot{\delta} / p_{\dot{\delta}}$, and use Euler integration with a time-step of $\Delta t = 0.2$ sec to obtain the affine system
\begin{subequations}
    \begin{align}
        x_{k+1} &= A x_{k} + B(p_{\dot{\delta}}) u_k + D w_k + r,
    \end{align}
    where, 
    \begin{align}
        A &= \begin{bmatrix}
            1 & 0 & 0 \\
            \frac{\nu_{x}} {\ell_f + \ell_r} \Delta t  & 1 & 0 \\
            \frac{\ell_r} {\ell_f + \ell_r} \nu_{x} \Delta t  & \nu_{x} \Delta t  & 1
        \end{bmatrix}, 
        ~~
        B(p_{\dot{\delta}}) = \begin{bmatrix}
            p_{\dot{\delta}} \Delta t \\
            p_{\dot{\delta}} \frac{\ell_r} {\ell_f + \ell_r} \Delta t \\
            0
        \end{bmatrix},
        \nonumber\\
        D &= \begin{bmatrix}
            \theta_{\delta} \Delta t & 0 & 0 \\
            0 & \theta_{\psi} \Delta t & 0 \\
            0 & 0 & \theta_{y} \Delta t
        \end{bmatrix},
        \quad
        r = \begin{bmatrix}
            0 \\
            - \sigma \nu_{x} \Delta t \\
            0 \\
        \end{bmatrix},
    \end{align}
\end{subequations}
where we have added the additive noise term to account for errors owing to the linear approximation.
% We set 
% $\nu_x = 1.0$, 
% $\sigma = 0.0$, 
% $\theta_{\delta} = \theta_{\psi} = \theta_{y} = 0.00001$,
% $\ell_f = 0.1$, $\ell_r = 0.9$,
% $Q_k = \diag(10000, 1, 1)$, $R_k = 1$ for $k=0, \dots, N-1$,
% % 
% $\bar{p} = 1.0$, $P = 0.05$,
% $\gamma = 0.1$,
% $\mu_I = [0, 0, 0.1]^\top$, $\Sigma_I = \bm{0}_{3 \times 3}$,
% $\mu_F = [0, 0, 0]^\top$, $\Sigma_F = \diag(0.01, 0.01, 0.0001)$,
% $N = 7$, 
% and $M = 20$.

% \begin{table}
%     \centering
%     \caption{Vehicle System and Problem Parameters}
%     \label{tab:params}
%     \begin{tabular}{|c|c|}\hline
%         Parameter & Value \\ \hline
%         $\ell_f$ & 0.1 \\ \hline
%         $\ell_r$ & 0.9 \\ \hline
%          & \\ \hline
%          & \\ \hline
%          & \\ \hline
%          & \\ \hline
%          & \\ \hline
%          & \\ \hline
%          & \\ \hline
%     \end{tabular}
% \end{table}

We compared the proposed approach \eqref{prob:d_cs_mc} with the method presented in \cite{knaup2023covariance}, which solves Problem~\eqref{prob:r_cs} using policy parameterization~\eqref{eq:affine_state_feedback}, as a baseline.
We used CasADi \cite{andersson2019casadi} and IPOPT \cite{wachter2006implementation} to formulate and solve the NLP given by Problem~\eqref{prob:d_cs_mc}.
The two methods are evaluated with additive Gaussian disturbances and with parameters drawn from the following distributions shown in Fig.~\ref{fig:dists}: Gaussian, uniform, Beta, and a mixture of two Gaussians. 
Table~\ref{tab:vehicle_trials} presents a comparison of the proposed approach (Adaptive Dual) vs. the baseline method \cite{knaup2023covariance} (Static Robust) for different parameter distributions and terminal covariance constraints.
The performance is measured as the average cost over 1,000 Monte Carlo trials and normalized using the average cost of the proposed method for the given problem data.
The terminal covariance is given by $\Sigma_F = \diag(0.01, \Sigma_F^{\theta}, \Sigma_F^{\theta})$.

\begin{figure}[h]
    \centering
    \includegraphics[width=\columnwidth]{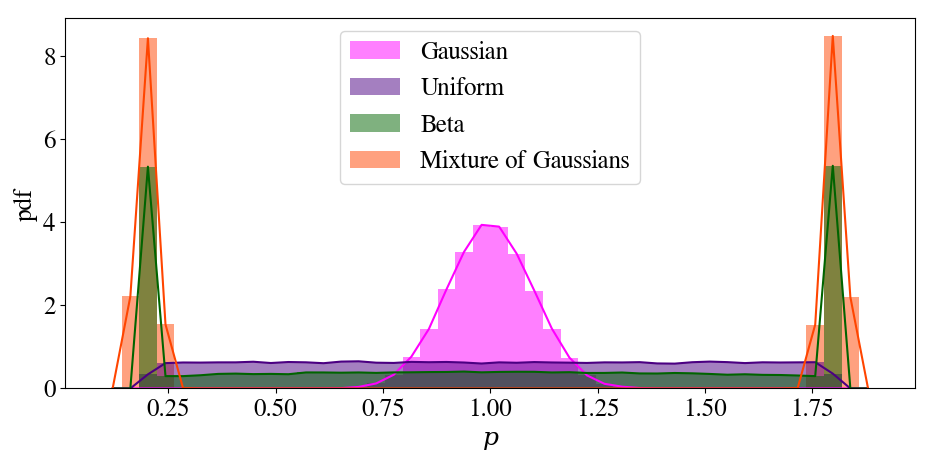}
    \caption{Parameter distributions.}
    \label{fig:dists}
\end{figure}

\begin{table}[h]
    \centering
    \caption{Average Cost Comparison for Various Parameter Distributions and Terminal Constraints}
    \label{tab:vehicle_trials}
    \begin{tabular}{|c||c|c|} \hline
         Control Formulation & Static Robust & Adaptive Dual \\ \hline
         Gaussian & & \\
         $\Sigma_F^{\theta} = 10^{-4}$ & $1.03$ & $1.0$ \\\hline
         Uniform & & \\
         $\Sigma_F^{\theta} = 10^{-3}$ & $1.41$ & $1.0$ \\\hline
         Uniform & & \\
         $\Sigma_F^{\theta} = 10^{-4}$ & $\infty$ & $1.0$ \\\hline
         Beta & & \\
         $\Sigma_F^{\theta} = 10^{-3}$ & $1.62$ & $1.0$ \\ \hline
         Beta & & \\
         $\Sigma_F^{\theta} = 10^{-4}$ & $\infty$ & $1.0$ \\ \hline
         Gaussian Mixture & & \\
         $\Sigma_F^{\theta} = 10^{-4}$ & $1.09$ & $1.0$ \\ \hline
    \end{tabular}
\end{table}

It may be seen in Table~\ref{tab:vehicle_trials} that the proposed adaptive dual covariance steering method is able to steer the state distribution to the prescribed terminal constraints while incurring a significantly lower average cost than the baseline static robust covariance steering method for the uniform and Beta distributions when $\Sigma_F^{\theta} = 10^{-3}$.
Moreover, when $\Sigma_F^{\theta} = 10^{-4}$, the proposed method is still able to meet the constraint using adaptive control, whereas the baseline static method is not able to find a solution that will satisfy the terminal constraint for the uniform and Beta distributions.
For the Gaussian distribution and mixture of Gaussians, the proposed adaptive dual covariance steering method still incurs a lower cost, but the benefits are less significant, and both the proposed and baseline methods are able to meet the terminal constraint when $\Sigma_F^{\theta} = 10^{-4}$.
These results may be interpreted through the entropy of the distributions shown in Fig.~\ref{fig:dists}. 
The Gaussian and bimodal mixture of Gaussians have less entropy and the parameter value is likely to lie around one or two values, respectively.
Thus, the static robust method is able to use the statistical information about these distributions to design a static control policy which statistically performs well for most parameter realizations.
However, the uniform and beta distributions have higher entropy, causing the advantages of an adaptive control policy to become more apparent.

Fig.~\ref{fig:trajectories_gaussian} shows sampled realized trajectories when a control policy generated by solving Problem~\eqref{prob:d_cs_mc} using the proposed method (blue) is compared with a baseline static control policy (red). 
The baseline static control policy is designed to robustly minimize the cost and satisfy the terminal constraints for parameters sampled from the known prior distribution.
The proposed method, on the other hand, also guarantees robust satisfaction of the terminal constraints, but additionally adapts the pre-computed control policy online in order to adjust to the changing parameter estimate and achieve a lower cost.
\begin{figure}[h]
    \centering
    \includegraphics[width=\columnwidth]{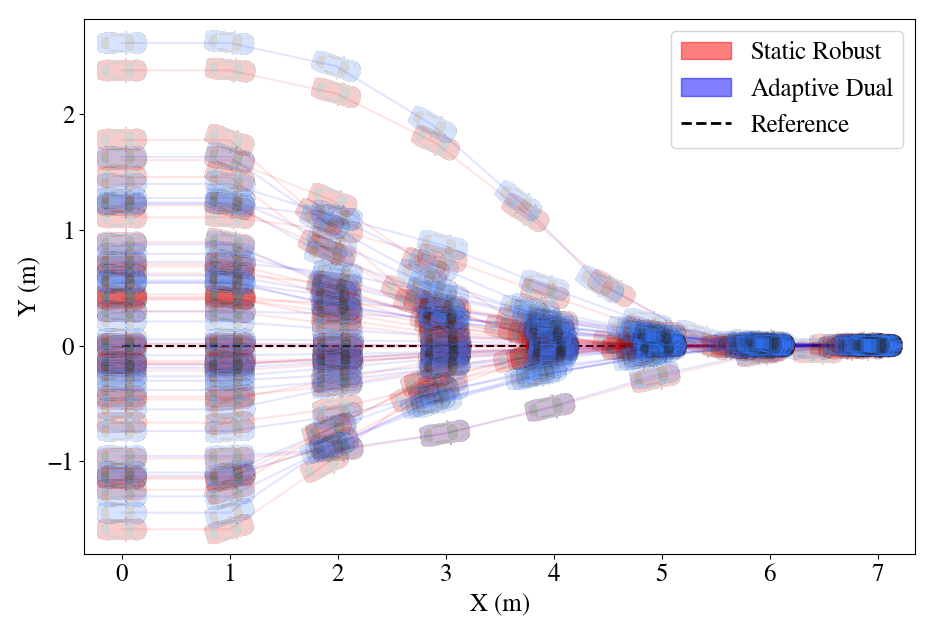}
    \caption{Sampled vehicle trajectories.}
    \label{fig:trajectories_gaussian}
\end{figure}
\begin{figure}[h]
    \centering
    \includegraphics[width=\columnwidth]{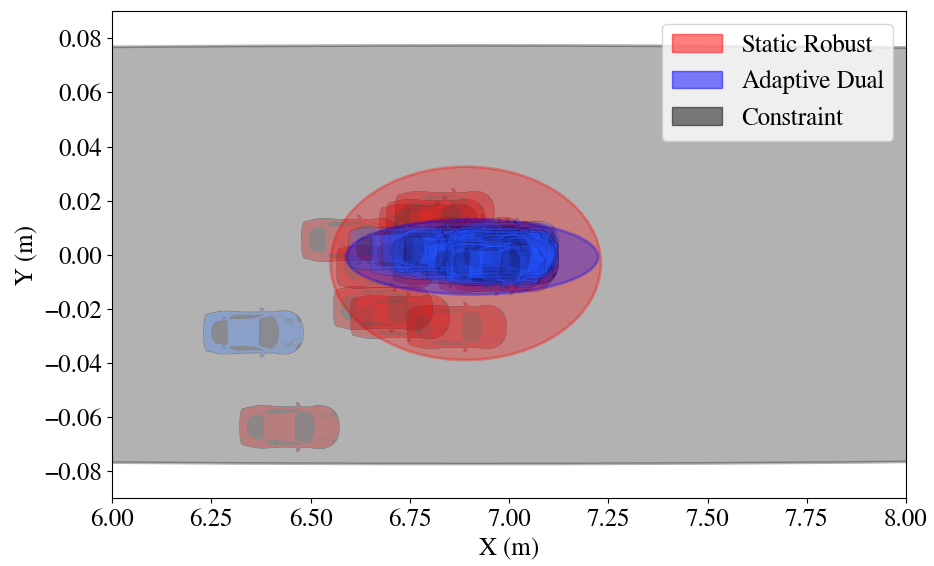}
    \caption{Terminal constraint visualization.}
    \label{fig:terminal_gaussian}
\end{figure}
% 
% \begin{figure}[h]
%     \centering
%     \includegraphics[width=\columnwidth]{figs/adaptive_control.png}
%     \caption{Average control effort Ccntribution from nominal and adaptive components.}
%     \label{fig:adaptive_control}
% \end{figure}
% This may be seen in Figs.~\ref{fig:trajectories_gaussian}-\ref{fig:terminal_gaussian}. 
Although the trajectories for both vehicles in Fig.~\ref{fig:trajectories_gaussian} appear similar,
as seen in Fig.~\ref{fig:terminal_gaussian}, the proposed adaptive dual covariance steering is able to successfully achieve smaller terminal covariance targets than the baseline static robust covariance steering by adapting the control policy to the parameter estimate online and, thus, achieves a lower average cost.

% The role of adaptive dual control may also be seen in Fig.~\ref{fig:adaptive_control}, which plots the average ratio of the control effort coming from the nominal---$v_k^{0}$, $L_k^{0}$---and adaptive---$v_k^{i}$, $L_{k}^{i}$---components for $i=1, \dots, n_p$ and $k = 0, \dots, N-1$. It may be seen that, while both the nominal and adaptive control components contribute equally when the parameter uncertainty is large, as the average parameter estimate error is reduced, the control policy shifts towards relying more heavily on the adaptive component, thus, exploiting the information gained from online parameter estimation.

\section{Conclusion}

In this paper, we formulated a dual control variation of the covariance steering problem for systems subject to unknown parameters. 
The adaptive dual covariance steering problem includes online parameter identification and a feedback policy parameterization which enables online adaptation to the estimated parameters while still ensuring constraint satisfaction in practice. 
We showed that the proposed method preserves the dual control property of encoding the control policy's effect on reducing model uncertainty through the parameter estimation.
We applied the proposed adaptive dual covariance steering approach to a vehicle control example and demonstrated it is able to outperform a baseline covariance steering method for systems subject to parametric uncertainty but which does not include adaptive dual control.

\bibliographystyle{IEEEtran}
\bibliography{references}

\end{document}